\documentclass[a4paper]{article}

\usepackage{amsmath,amssymb,amsthm,amsfonts}
\numberwithin{equation}{section}
\usepackage{physics}
\usepackage[english]{babel}
\usepackage[margin=3cm]{geometry}
\usepackage{csquotes}
\usepackage{mathtools}
\usepackage{relsize}
\usepackage{indentfirst}
\usepackage{thm-restate}
\usepackage{mathrsfs}
\usepackage[dvipsnames]{xcolor}
\usepackage{paracol}
\usepackage{enumitem}
\usepackage{soul}
\usepackage{aligned-overset}
\usepackage[thinc]{esdiff}
\usepackage{upgreek}
\usepackage{datetime}
\usepackage{tensor}
\usepackage{caption}
\usepackage{subcaption}
\usepackage{appendix}
\allowdisplaybreaks

\usepackage[dvipsnames]{xcolor}
\usepackage{hyperref}	
\newcommand\myshade{85}
\colorlet{mylinkcolor}{violet}
\colorlet{mycitecolor}{YellowOrange}
\colorlet{myurlcolor}{Aquamarine}
\hypersetup{
    linkcolor  = blue!\myshade!black,
    citecolor  = mycitecolor!\myshade!black,
    urlcolor   = myurlcolor!\myshade!black,
    colorlinks = true,
    hypertexnames=false,
}
\usepackage[capitalize,nameinlink]{cleveref}\usepackage[
style = alphabetic,
maxbibnames=9,maxcitenames=9,
doi=true,isbn=true,giveninits=true,
backend=biber,
url=false,eprint=true,arxiv=abs]{biblatex}
\theoremstyle{plain}

\theoremstyle{plain}

\Crefname{assumption}{assumption}{assumptions}

\newtheorem{lemma}{Lemma}[section]
\newtheorem{corollary}{Corollary}[section]
\newtheorem{theorem}{Theorem}[section]

\usepackage{pageslts}
\usepackage{fancyhdr}
\DeclarePairedDelimiterX{\cgip}[2]{\langle}{\rangle}{#1,#2}
\newcommand{\CC}{\mathbb{C}}
\newcommand{\RR}{\mathbb{R}}

\newcommand{\gc}[1]{\abs{\g{#1}}}
\newcommand{\g}[1]{G_{#1}}

\title{A simple proof of the group-theoretic Zhang--Yeung inequality}
\usepackage{authblk}

\author{Harold Nieuwboer\thanks{hani@math.ku.dk} }
\author{Lubashan Pathirana\thanks{lpk@math.ku.dk}}
\affil{Department of Mathematical Sciences and QMATH, University of Copenhagen, Denmark}
\date{}
\begin{document}
\pagenumbering{arabic}
\lhead{\thepage}
\maketitle
\vspace{-1cm}

\begin{abstract}
  Zhang and Yeung (IEEE Trans. Inf. Theory, 1998) established the first non-Shannon-type inequality that holds for all entropic vectors.
  Chan and Yeung (IEEE Trans. Inf. Theory, 2002) showed that there is a one-to-one relation between linear entropy inequalities and multiplicative inequalities involving cardinalities of subgroups of a finite group.
  The Shannon inequality admits a simple proof in the group theoretic setting, but a direct proof of the translation of Zhang--Yeung's inequality to the group-theoretic setting remained elusive.
  We resolve this open problem, giving a direct proof of the group-theoretic Zhang--Yeung inequality for cardinalities of subgroups of finite groups, using elementary group-theoretic and counting arguments.
\end{abstract}

\section{Introduction}

One of the pivotal results of classical information theory is Shannon's inequality~\cite{shannonMathematicalTheoryCommunication1948}:
for (discrete) random variables~$A,B,C$, their entropy satisfies the linear inequality
\[
  H(A B) + H(B C) \geq H(A B C) + H(B),
\]
also known as the submodularity inequality.
Equivalently, the conditional mutual information~$I(A:C | B) = H(A B) + H(B C) - H(A B C) - H(B)$ is non-negative.
A natural question, posed by Pippenger~\cite{pippenger1986laws}, is: which other relations, if any, hold between entropies of joint random variables?

This question may be formalized as follows.
Let~$n \geq 1$, and suppose we are given jointly distributed finite-valued discrete random variables~$X_1, \dotsc, X_n$.
For a subset~$I = \{i_1, \dotsc, i_k\} \subseteq [n]$ we denote by~$H(X_{i_1}, \dotsc, X_{i_k})$ the Shannon entropy of the joint random variable~$(X_{i_1}, \dotsc, X_{i_k})$, sometimes abbreviated~$H(X_I)$. We additionally set~$H(X_\emptyset) = 0$.
The entropic vector associated with the~$X_i$ is then defined as the vector in~$\RR^{2^n}$, indexed by subsets of $[n]$, with~$I$-th entry equal to~$H(X_I)$.
The set of all entropic vectors is denoted by~$\Gamma_n^*$.
Pippenger's question then reduces to characterizing all valid inequalities for the set~$\Gamma_n^*$.

Zhang and Yeung~\cite{zhangNonShannontypeConditionalInequality1997} showed that~$\Gamma_n^*$ is itself not closed when~$n \geq 3$, in the sense that there are non-entropic vectors which can be arbitrarily well approximated by entropic vectors.
However, its topological closure~$\overline{\Gamma_n^*}$ is a convex cone, and hence characterized by~\emph{linear} inequalities.
Its relative interior is contained in~$\Gamma_n^*$~\cite[Thm.~1]{matusTwoConstructionsLimits2007}, so the closure is only relevant for the boundary points.

For~$n = 3$, Zhang and Yeung showed that $\overline{\Gamma_3^*}$ is indeed determined by the Shannon-type inequalities: monotonicity, and submodularity.
For~$n \geq 4$, the situation drastically changes~\cite{zhangCharacterizationEntropyFunction1998}:
there are valid linear inequalities for~$\overline{\Gamma_n^*}$ that cannot be derived from Shannon-type inequalities:
\begin{theorem}
  \label{thm:zy inequality}
  For finite-outcome discrete random variables~$X_1, X_2, X_3, X_4$, one has the inequality
  \[
2 I(X_3:X_4)  \leq I(X_1 : X_2) + I(X_1 : X_3 X_4) + 3 I(X_3:X_4|X_1) + I(X_3:X_4|X_2).
     \tag{ZY}
     \label{eq:ZY}
  \]
  Moreover, this inequality cannot be derived from Shannon-type inequalities.
\end{theorem}
Characterizing the full set of inequalities has remained an open problem since then.
By this point, many more inequivalent inequalities have been found~\cite{makarychevNewClassNonShannontype2002,zhangNewNonShannonType2003,doughertySixNewNonShannon2006,matusInfinitelyManyInformation2007,xuProjectionMethodDerivation2008,doughertyNonShannonInformationInequalities2011,csirmazBookInequalities2014,csirmazInformationInequalitiesFive2026}, including infinite families that suffice to show that~$\overline{\Gamma_n^*}$ is not polyhedral~\cite{matusInfinitelyManyInformation2007} (later observed to come from a single non-linear inequality~\cite{chanNonlinearInformationInequalities2008}).
A full description of the entropy cone remains elusive for now. 
For a better overview of connections and applications, we refer the reader to~\cite{yeungFacetsEntropy2015}.

Chan~and~Yeung~\cite{chanRelationInformationInequalities2002} observed that inequalities on Shannon entropies directly translate to inequalities on cardinalities of subgroups of a finite group~$G$.
The construction is as follows.
Consider subgroups~$G_1,\dotsc,G_n \leq G$, and for~$I \subseteq [n]$ write~$G_I = \cap_{i \in I} G_i$ (setting~$G_\emptyset = G$).
One can now construct random variables~$X_1, \dotsc, X_n$, by simply uniformly sampling~$g \in G$, and setting~$X_i = g G_i \in G/G_i$.
The entropy of~$X_i$ is~$H(X_i) = \log (\abs{G} / \abs{G_i})$, as the coset~$g G_i$ is sampled uniformly among all cosets.
More generally, a tuple~$(g G_{i_1}, \dotsc, g G_{i_k})$ is in one-to-one correspondence with the left-coset~$g G_I$ of the intersection~$G_I$, and therefore~$X_I$ corresponds to uniformly sampling from the~$G_I$-cosets; hence $H(X_I) = \log(\abs{G} / \abs{G_I})$.
Any (linear) inequality on entropies of arbitrary jointly distributed random variables therefore translates to a (multiplicative) inequality on cardinalities of subgroups and their intersections.
The main result of~\cite{chanRelationInformationInequalities2002} is that the converse also holds: multiplicative inequalities on cardinalities of subgroups imply linear inequalities on the entropy cone.
The underlying geometric reason is that the closed convex cone generated by the \emph{group-characterizable} entropy vectors is exactly~$\overline{\Gamma_n^*}$.

As an example, the Shannon inequality may be cast in group theoretic language as follows.
Let~$G_1, G_2, G_3 \leq G$.
Then
\[
  \log \frac{\gc{}}{\gc{12}} + \log \frac{\gc{}}{\gc{23}} \geq \log \frac{\gc{}}{\gc{123}} + \log \frac{\gc{}}{\gc{2}},
\]
or equivalently
\[
  \gc{12} \gc{23} \leq \gc{123} \gc{2}.
\]
This admits the following elementary group-theoretic proof.
The set~$\g{12} \g{23} = \{ g g' : g \in \g{12},\, g'\in\g{23} \}$ is clearly contained in~$\g{2}$, so~$\abs{\g{12} \g{23}} \leq \gc{2}$.
For every~$z \in \g{12} \g{23}$, there are exactly~$\gc{123}$ pairs~$(g, g') \in \g{12} \times \g{23}$ such that~$g g' = z$.
Indeed, suppose~$(h, h') \in \g{12} \times \g{23}$ is another such pair: then~$d = h^{-1} g = h' (g')^{-1}$ is in~$\g{12}$ and~$\g{23}$, hence in~$\g{123}$.
Moreover, $(h, h')$ can be recovered from~$d$ and~$(g,g')$ via~$h = g d^{-1}$, $h' = d g'$.
Therefore~$\g{123}$ is in bijection with the set of pairs in~$\g{12} \times \g{23}$ whose product is~$z$.
We conclude that
\[
  \gc{123} \abs{\g{12} \g{23}} = \gc{12} \gc{23}.
\]

Similarly, the Zhang--Yeung inequality (\ref{eq:ZY}) translates to the following inequality on subgroup cardinalities:
\begin{theorem}
  \label{thm:zcy inequality}
  For every finite group~$G$ and subgroups~$G_1, G_2, G_3, G_4 \leq G$, one has the inequality
  \[
    \gc{34}^3 \gc{13}^3 \gc{14}^3 \gc{23} \gc{24} 
    \leq
    \gc{1} \gc{3}^2 \gc{4}^2 \gc{12} \gc{134}^4 \gc{234}
    \tag{G-ZY}
    \label{eq:G-ZY}
  \]
\end{theorem}
Giving a group-theoretic interpretation of the Zhang--Yeung inequality was stated as an open problem by Chan--Yeung~\cite[Ex.~5.2]{chanRelationInformationInequalities2002}, and giving a direct proof was recently pointed out again as remaining open by Yeung~\cite{yeungInequalitiesRevisited2025}.
We resolve this open problem by giving a short elementary proof.

The proof we give does not appear to be a simple translation of the usual proof, which relies on the \emph{copy lemma}. This lemma featured in the original proof of the Zhang--Yeung inequality \cite{zhangCharacterizationEntropyFunction1998} and was later extracted as a separate tool by Dougherty, Freiling and Zeger~\cite{doughertySixNewNonShannon2006}.
Consider four jointly distributed random variables~$X_1,X_2,X_3,R$.
Then~$R$ is called an~$X_3$-copy of~$X_2$ over~$X_1$ if $(X_1,R)$ and~$(X_1,X_2)$ have the same marginal probability distribution, and~$I(X_2 X_3 : R | X_1) = 0$.
The last condition translates to independence of~$R$ from~$(X_2, X_3)$, conditioned on~$X_1$.
The copy lemma constructs such an~$R$ explicitly for arbitrary~$(X_1, X_2, X_3)$.
The Zhang--Yeung inequality, and all other known  unconditional linear discrete Shannon entropy inequalities (that the authors are aware of), admit proofs by first applying the copy lemma some number of times and subsequently applying Shannon-type inequalities~\cite{kacedEquivalenceTwoProof2013a,csirmazExploringEntropicRegion2025}.

A ``group-theoretic'' copy construction may be carried out as follows.
Let~$\g1, \g2, \g3 \leq G$ be subgroups of a finite group.
Define random variables~$X_1, X_2, X_3$ by sampling~$g \in G$ from the uniform distribution and then setting~$X_i = g G_i$.
To obtain an~$R$ which is an~$X_3$-copy of~$X_2$ over~$X_1$, we change the underlying sample space to~$\Omega_1 = \{ (g, h) \in G \times G : g G_1 = h G_1 \}$.
We redefine our random variables by sampling~$(g,h)$ uniformly from~$\Omega_1$, setting~$X_i = g G_i$ and~$R = h G_2$.
One may verify that indeed~$(X_1,X_2)$ and~$(X_1,R)$ have the same distribution:
indeed~$(g G_1, g G_2)$ and~$(g G_1, h G_2) = (h G_1, h G_2)$ have equal distributions since the marginal distributions on~$g$ and~$h$ themselves are uniform.
Moreover, the construction embodies the idea of conditional independence: first sample~$g \in G$ uniformly, witness~$X_1 = g G_1$, uniformly sample~$h$ from~$g G_1$, and set~$R = h G_2$.
This interpretation easily leads to the conclusion that~$I(X_2 X_3 : R | X_1) = 0$.

This copy construction has some undesirable properties.
First of all, we immediately leave the setting considered by Chan and Yeung, as the sample space~$\Omega_1$ is no longer a group in general: it is a subgroup of~$G \times G$ if and only if~$G_1$ is a normal subgroup of~$G$.
One cannot reduce to this setting without loss of generality:
If the subgroups involved are normal, then the associated entropy vector lies in a smaller cone~\cite[Thm.~6]{lindenQuantumEntropyCone2013}. In particular, even if only~$G_1$ is normal, the Ingleton inequality is valid~\cite[App.~B]{maoIngletonViolatingFiniteGroups2017a} (as the condition~$G_1 G_2 = G_2 G_1$ is automatically satisfied), but this inequality does not hold for arbitrary entropic vectors.
Moreover, iterative copying of tuples of variables, such that among the variables are those which are themselves already copies of other variables, naturally leads one to consider non-uniform probability distributions on the sample space.

The proof presented in this paper does not appear to rely on such a copy construction.
The tools used are (1) exact counting of the fiber sizes of the multiplication between two subgroups, and (2) elementary linear algebraic results on the group algebra~$\CC[G]$.
We hope that the argument leads to new tools and inspiration for proving other, potentially new, entropy inequalities.

\section{Preliminaries}
\label{sec:prelims}
The first tool we need is the following lemma.
\begin{lemma}
  \label{lemma:basic_counting_strong}
  Let $H,K\leq G$, and define
  \[
    \mu\colon H\times K\longrightarrow G,
    \qquad
    \mu(h,k):=hk.
  \]
  For every $z \in HK$,
  \[
    |\mu^{-1}(\{z\})|=|H\cap K|.
  \]
  Consequently, for every \(X\subseteq G\) (\(X\) is \emph{any} subset, not necessarily a subgroup),
  \[
    |\mu^{-1}(X)|
    =
    |H\cap K|\,|HK\cap X|.
  \]
\end{lemma}
\begin{proof}
  Fix \(z=h_0k_0\in HK\). 
  The map
  \[
    H\cap K\longrightarrow\mu^{-1}(\{z\}),
    \qquad
    t\longmapsto(h_0t,t^{-1}k_0),
  \]
  is a bijection:
  (Suppose \((h_0t,t^{-1}k_0) = (h_0t',(t')^{-1}k_0)\) for some $t,t' \in H \cap K$. Then \(h_0t = h_0t'\implies t = t'\).
  Now let $(h,k)\in \mu^{-1}(\{z\})$. Then $hk = z  =h_0k_0$. Thus take $t = (h_0)^{-1}h = k_0k^{-1} \in H \cap K$ and then 
  the map above maps $t$ to $(h, k)$).
  This gives the first assertion that $|\mu^{-1}(\{z\})| = |H \cap K|$. 
  Now, for $X\subseteq G$, we have
  \[
    \mu^{-1}(X) 
    = 
    \bigsqcup_{z\in HK \cap X} \mu^{-1}(\{z\})
  \]
  and thus summing over \(z\in HK\cap X\) proves the second assertion.
\end{proof}

We will also need some linear algebraic facts about the group algebra~$\CC[G]$.
For a subgroup~$H \leq G$, define the \emph{averaging operator}
\begin{equation*}
  \Phi_H = \frac{1}{\abs{H}} \sum_{h \in H} h,
\end{equation*}
viewed as a linear map~$\CC[G] \to \CC[G]$ (acting by left multiplication).
We endow~$\CC[G]$ with the inner product
\[
  \cgip{g}{h} = \begin{cases}
    1 & \text{if } g = h, \\
    0 & \text{otherwise.}
  \end{cases}
\]
\begin{lemma}
  \label{lem:averaging operator projector}
  Let~$H \leq G$ be a subgroup.
  Then~$\Phi_H$ is a self-adjoint operator on~$\CC[G]$ and is the orthogonal projection onto~$\CC[G]^H$, with~$\Tr[\Phi_H] = \abs{G} / \abs{H}$.
\end{lemma}
\begin{proof}
  For every~$g, g' \in G$, we have
  \[
    \cgip{g}{\Phi_H g'} 
    = \frac{1}{\abs{H}} \sum_{h \in H} \cgip{g}{h g'}
    = \frac{1}{\abs{H}} \sum_{h \in H} \cgip{h^{-1} g}{g'}
    = \cgip{\Phi_H g}{g'}
  \]
  as~$H$ is closed under inversion.
  Therefore~$\Phi_H$ is self-adjoint.
  Similarly~$\Phi_H^2 = \Phi_H$, as every~$h \in H$ factorizes as~$h = h' h''$ for exactly~$\abs{H}$ many pairs~$(h', h'') \in H \times H$.
  Moreover, $h \Phi_H v = \Phi_H v$ for every~$v \in \CC[G]$ and so the image is contained in~$\CC[G]^H$; furthermore every~$v \in \CC[G]^H$ clearly satisfies~$\Phi_H v = v$.
  To compute the trace of~$\Phi_H$, note that~$\Tr[\Phi_H] = \sum_{g \in G} \cgip{g}{\Phi_H g} = \sum_{g \in G} 1/\abs{H} = \abs{G}/\abs{H}$, where the second step follows from the fact that~$g = h g$ if and only if~$h$ is the identity element of~$G$.
\end{proof}
\begin{lemma}
  \label{lem:averaging operator monotone}
  Let~$H \leq K \leq G$ be subgroups.
  Then~$\Phi_H \succeq \Phi_K$.
\end{lemma}
\begin{proof}
  As both operators are projectors (\cref{lem:averaging operator projector}), it suffices to observe that~$\CC[G]^K \subseteq \CC[G]^H$, which is trivial: every vector in~$\CC[G]$ which is invariant under left-multiplication by~$K$ is also invariant under left-multiplication by~$H$.
\end{proof}
\begin{lemma}
  \label{lem:trace computation}
  Let~$H,K,L,M \leq G$ be subgroups.
  Then
  \[
    \frac{1}{\abs{G}} \Tr[\Phi_K \Phi_H \Phi_L \Phi_M] = \frac{\abs{H K \cap L M}}{\abs{H K} \abs{L M}}.
  \]
\end{lemma}
Note that the ordering of the subgroups~$H,K$ differs on the left- and right-hand sides.
\begin{proof}
  Observe that
  \[
    \Tr[\Phi_K \Phi_H \Phi_L \Phi_M] = \frac{1}{\abs{K}\abs{H}\abs{L}\abs{M}} \sum_g \sum_k \sum_h \sum_l \sum_m \cgip{g}{k h l m g}.
  \]
  Each term is~$1$ iff~$k h l m = e$ iff~$l m = h^{-1} k^{-1}$, and zero otherwise.
  We count the number of these tuples as follows.
  Fix~$z \in L M$.
  Then by \cref{lemma:basic_counting_strong} there are~$\abs{L \cap M}$ pairs $(l,m) \in L \times M$ such that~$l m = z$.
  Simultaneously, if~$z \in H K$, there are~$\abs{H \cap K}$ pairs~$(h,k)$ such that~$h^{-1} k^{-1} = z$, and~$0$ such pairs if~$z \not\in H K$.
  Therefore for each~$z \in H K \cap L M$, we have~$\abs{L \cap M} \abs{H \cap K}$ many tuples~$(k,h,l,m)$ such that~$h^{-1} k^{-1} = z = l m$.
  The result now follows from using~$\abs{H K} = \abs{H} \abs{K} / \abs{H \cap K}$ and similarly for~$L, M$.
\end{proof}
Combining this with the operator monotonicity from~\cref{lem:averaging operator monotone}, and the standard inequality~$\Tr[P Q] \geq 0$ for positive semidefinite~$P,Q$, leads to the following two corollaries:
\begin{corollary}
  \label{cor:pivot monotonicity}
  Let~$H,K,L,M \leq G$ be subgroups such that~$L \leq M$.
  Then
  \[
    \frac{\abs{H K \cap K L}}{\abs{H K} \abs{K L}} \geq \frac{\abs{H K \cap K M}}{\abs{H K} \abs{K M}}.
  \]
\end{corollary}
\begin{proof}
  By~\cref{lem:averaging operator monotone} we have~$\Phi_M \preceq \Phi_L$ and hence~$\Tr[\Phi_K \Phi_H \Phi_K \Phi_L] \geq \Tr[\Phi_K \Phi_H \Phi_K \Phi_M]$.
\end{proof}
\begin{corollary}
  \label{cor:cauchy schwarz lower bound}
  Let~$H,K,L \leq G$ be subgroups such that~$H,K \leq L$.
  Then
  \[
    \frac{\abs{H K \cap K H}}{\abs{H K} \abs{K H}} \geq \frac{1}{\abs{L}}.
  \]
\end{corollary}
\begin{proof}
  The operators~$\Phi_K \Phi_H \Phi_K$ and~$\Phi_L \Phi_K \Phi_L$ are positive semidefinite, hence~\cref{lem:averaging operator monotone} yields
  \begin{align*}
    \abs{G} \frac{\abs{H K \cap K H}}{\abs{H K} \abs{K H}} & = \Tr[\Phi_K \Phi_H \Phi_K \Phi_H] 
    \geq \Tr[\Phi_K \Phi_H \Phi_K \Phi_L]
    \geq \Tr[\Phi_K \Phi_L \Phi_K \Phi_L]
  \end{align*}
  and similarly~$\Tr[\Phi_K \Phi_L \Phi_K \Phi_L] \geq \Tr[\Phi_L \Phi_L \Phi_K \Phi_L] \geq \Tr[\Phi_L \Phi_L \Phi_L \Phi_L] = \Tr[\Phi_L] = \gc{} / \abs{L}$.
\end{proof}

\section{The proof of the Zhang--Yeung inequality}
We start with a simple lemma where we count the elements of a single set in two different ways.
\begin{lemma}
  \label{lem:E two countings}
  Consider the set 
  \[
    E = \{ (a,b,c,d) \in \g{13} \times \g{14} \times \g{23} \times \g{24} : bc \in \g3 \g4 \}
  \]
  The cardinality of~$E$ satisfies
  \[
    \abs{E} = \gc{13} \gc{1234} \abs{\g{14} \g{23} \cap \g3 \g4} \gc{24}  \leq \abs{\g1 \g2 \cap \g3 \g4} \gc{12} \gc{134} \gc{234}.
  \]
\end{lemma}
\begin{proof}
  The exact cardinality of~$E$ is determined by applying~\cref{lemma:basic_counting_strong} to the multiplication map~$G_{14} \times G_{23} \to G$, concluding that the size of the preimage of~$\g3 \g4$ is~$\gc{1234} \abs{\g{14} \g{23} \cap \g3 \g4}$.
  To prove the upper bound, we proceed as follows.
  Observe that for any~$(a,b,c,d) \in E$, we have~$ab \in G_1$, $cd \in G_2$, and~$abcd = a(bc)d \in G_3 G_4$, as well as~$abcd \in G_1 G_2$.
  Therefore,
  \begin{align*}
    E & = \coprod_{z \in \g1 \g2 \cap \g3 \g4} \{ (a,b,c,d) \in E : abcd = z \} \\
      & = \coprod_{z \in \g1 \g2 \cap \g3 \g4} \coprod_{(p,q) \in \g1 \times \g2}  \{ (a,b,c,d) \in E : pq = z, ab = p, cd = q \}.
  \end{align*}
  For fixed~$z \in G$, there are at most~$\gc{12}$ pairs~$(p,q)$ in~$G_1 \times G_2$ such that~$pq = z$, and for each fixed~$p$ and~$q$ respectively, there are at most~$\gc{134}$ and~$\gc{234}$ pairs~$(a,b) \in \g{13} \times \g{14}$ and~$(c,d) \in \g{23} \times \g{24}$ such that~$ab = p$ and~$cd = q$.
  Therefore~$\abs{E} \leq \abs{\g1 \g2 \cap \g3 \g4} \gc{12} \gc{134} \gc{234}$.
\end{proof}
This already has much of the structure of~\cref{thm:zcy inequality}, as one can upper bound~$\abs{\g1 \g2 \cap \g3 \g4} \leq \abs{\g3 \g4} = \gc3 \gc4 / \gc{34}$.
It remains to prove good lower bounds on~$\abs{\g{14} \g{23} \cap \g3 \g4}$.
Note that this quantity is not directly related to the cardinality of a subgroup of~$G$, so it is not clear how to interpret this in terms of e.g., the conditional mutual information between some random variables arising from the Chan--Yeung construction.

We first change the number of occurrences of the subgroup~$G_2$ through the following lemma:
\begin{lemma}
  \label{lem:first lower bound}
  We have the inequality
  \[
    \frac{\abs{\g{14} \g{23} \cap \g3 \g4} \gc{1234}}{\gc{14} \gc{23}}
    =
    \frac{\abs{\g{14} \g{23} \cap \g3 \g4}}{\abs{\g{14} \g{23}}}
    \geq
    \frac{\abs{\g{14} \g3 \cap \g3 \g4}}{\abs{\g3 \g4}}.
  \]
\end{lemma}
\begin{proof}
  First, we have the trivial inequality~$\abs{\g{14} \g{23} \cap \g3 \g4} \geq \abs{\g{14} \g{23} \cap \g3 \g{14}}$.
  We complete the proof by two applications of~\cref{cor:pivot monotonicity}:
  \[
    \frac{\abs{\g{14} \g{23} \cap \g3 \g{14}}}{\abs{G_{14} G_{23}}} \geq 
    \frac{\abs{\g{14} \g{3} \cap \g3 \g{14}}}{\abs{G_{14} G_{3}}} \geq 
    \frac{\abs{\g{14} \g{3} \cap \g3 \g{4}}}{\abs{G_{4} G_{3}}}.
    \qedhere
  \]
\end{proof}

\begin{lemma}
  \label{lem:second lower bound}
  We have the inequality
  \[
    \abs{\g{14} \g3 \cap \g3 \g4} \geq \frac{\abs{\g{13} \g{14}}^2 \gc{34}}{\gc1 \gc{134}}
  \]
\end{lemma}
\begin{proof}
  Consider~$\mathcal{R} \coloneqq \g{14} \g{13} \cap \g{13} \g{14} \subseteq \g1$.
  Then~$\mathcal{R} \g{34} \subseteq \g{14} \g3 \cap \g3 \g4 \eqqcolon \mathcal{A}$.
  Applying~\cref{lemma:basic_counting_strong} to the multiplication map~$\nu_{1,34}\colon \g1 \times \g{34} \to G$ yields
  \[
    \abs{\nu_{1,34}^{-1}(\mathcal{A})} = \abs{\g{134}} \abs{\g1 \g{34} \cap \mathcal{A}} \leq \gc{134} \abs{\mathcal{A}}.
  \]
  The left-hand side is lower bounded by~$\abs{\mathcal{R}} \gc{34}$, and by~\cref{cor:cauchy schwarz lower bound} we have
  \[
    \abs{\mathcal{R}} \geq \frac{\abs{\g{13} \g{14}}^2}{\gc1}.
  \]
  Rearranging the inequalities yields
  \[
    \abs{\mathcal{A}}
    \geq
    \frac{\abs{\g{13} \g{14}}^2 \gc{34}}{\gc1 \gc{134}}.
    \qedhere
  \]
\end{proof}
\begin{proof}[Proof of~\cref{thm:zcy inequality}]
  By applying~\cref{lem:E two countings,lem:first lower bound,lem:second lower bound} in succession, we obtain
  \begin{align*}
    \abs{\g1 \g2 \cap \g3 \g4} \gc{12} \gc{134} \gc{234} & \geq \abs{E} = \gc{13} \gc{1234} \abs{\g{14} \g{23} \cap \g3 \g4} \gc{24} \\
                                                         & \geq \gc{13} \gc{24} \gc{14} \gc{23} \frac{\abs{\g{14} \g3 \cap \g3 \g4}}{\abs{\g3 \g4}} \\
                                                         & \geq \frac{\gc{13} \gc{24} \gc{14} \gc{23} \abs{\g{13} \g{14}}^2 \gc{34}}{\abs{\g3 \g4} \gc1 \gc{134}} \\
                                                         & = \frac{\gc{13}^3 \gc{24} \gc{14}^3 \gc{23} \gc{34}^2}{\gc3 \gc4 \gc1 \gc{134}^3}
  \end{align*}
  which becomes the group-theoretic Zhang--Yeung inequality~(\ref{eq:G-ZY}) after estimating~$\abs{\g1 \g2 \cap \g3 \g4} \leq \abs{\g3 \g4} = \gc3 \gc4 / \gc{34}$.
\end{proof}

\vspace{1em}
\noindent
\textbf{AI Usage Declaration.}
No AI tools were involved in finding the results of the paper.
AI tools were used to assist in literature checking.
All text is written by the authors, and any omissions or inaccuracies are entirely the responsibility of the authors.

\vspace{1em}
\noindent
\textbf{Acknowledgements.}
The authors would like to thank Thomas C. Fraser for bringing this question to their attention, as well as for interesting discussions.
HN is supported by the European Union's Horizon Europe research and innovation programme under the Marie Sk\l{}odowska-Curie Actions (MSCA) Postdoctoral Fellowship, Grant Agreement No. 101212204 (AsympTensorPolytope).
LP acknowledges the Danish e-Infrastructure Consortium (DeiC) 5260-00014B grant, which supported part of this work.
HN and LP also acknowledge support from Villum Fonden via the QMATH Centre of Excellence (Grant No.~10059).

\printbibliography

\end{document}